%% file: full-final.tex
\documentclass[runningheads,envcountsect,envcountsame]{llncs}
\usepackage{url}
\usepackage{hyperref}
\usepackage[T1]{fontenc}
\usepackage[utf8]{inputenc}
\usepackage{xcolor}
\usepackage{blindtext}
\usepackage{hyperref}
\usepackage{mathcommand}
\usepackage{amssymb,amsfonts,amsmath,mdwlist,verbatim}
\usepackage{cleveref}
\usepackage{marginnote}
\usepackage{paralist}
\usepackage{tikz}
\usetikzlibrary{positioning}
 \usetikzlibrary{calc}

 \spnewtheorem{mainrem}[theorem]{Remark}{\bfseries}{\itshape}

\newcommand{\MSO}{\mathit{MSO}}

\def\nat{\mathbb{N}}

\def\type{\mathit{type}}

\usepackage{mathtools}   

\def\vp{\varphi}

\def\mM{{\mathcal{M}}}

\def\type{{\mathit{type}}}

\newrobustcmd\alex[1]{\textcolor{teal}{A: #1}}

\newrobustcmd\thomas[1]{\textcolor{violet}{T: #1}}

\usepackage{cleveref,thm-restate}

\def\vp{\varphi}

\newcommand{\cC}{\mathcal{C}}

\newcommand{\qr}{{\mathrm{qr}}}
\newcommand{\Formula}[2]{{\mathfrak{Form}^{#1}_{#2}}}
\newcommand{\Hint}[2]{H^{#1}_{#2}}

\newcommand{\Typ}[2]{{\mathrm{type}^{#1}_{#2}}}

\newcommand{\frakB}{\mathfrak{B}}

\usepackage{xparse}

\newcommand{\Nat}{\mathbb{N}}

\newcommand{\frakA}{\mathfrak{A}}
\title{Decidability of MSO Reparameterization over Countable   Chains}

\author{Alexander Rabinovich\orcidID{0000-0002-1460-2358}}

\institute{Tel Aviv University, Israel\\
\email{rabinoa@tauex.tau.ac.il}, \\
\url{https://www.tau.ac.il/~rabinoa/}} 

\begin{document}
\maketitle
\begin{abstract}
Interpretations are a fundamental tool in mathematical logic, 
allowing structures to be encoded within other structures via logical definitions.
We study $\MSO$ \emph{multidimensional point interpretations}, where elements of an
interpreted structure are represented by tuples of elements of the host
structure, and address the problem of simplifying such interpretations
by reducing their representation dimension.

To formalize simplification, we use the notion of
\emph{reparameterization}.     
Our main result shows that, over the class of countable labelled linear
orders, it is decidable whether a given $\MSO$ formula admits a
$d$-dimensional reparameterization.
As a consequence, every interpretation whose domain admits such a
reparameterization is equivalent to a $d$-dimensional point
interpretation.

\end{abstract}

\section{Introduction}

The notion of interpretation was first carefully defined and developed in
the foundational work of Tarski, Mostowski, and Robinson~\cite{TMR53}.
Interpretations are a fundamental tool in mathematical logic, playing a central role in the foundations of mathematics, as well as in the philosophy of science.

Consider Monadic Second-Order\footnote{Let us recall that  monadic second-order logic is an extension of first-order logic by set variables 
 (see   \Cref{sect:mso} for a formal definition).}  ($\MSO$) $d$-dimensional \emph{(point) interpretations} of a structure
\(\mathfrak A\) in a structure \(\mathfrak B\).
Under such an interpretation, each element of $\mathfrak{A}$ is represented within $\mathfrak{B}$ by a $d$-tuple of elements of $\mathfrak{B}$.

A central problem is whether a given interpretation can be simplified, for
instance by reducing the dimension of the representation.
The domain of the interpreted structure is specified by an $\MSO$ formula
\(\varphi( {x_1},\dots,x_m)\) over \(\mathfrak B\).
To formalize this notion of simplification, we use the notion of
\emph{reparameterization}, introduced in~\cite{GallotLhoteNguyen:PolyGrowth:2025}
for $\MSO$-definable relations.

Throughout, $\bar x$ and $\bar y$ denote tuples of first-order variables.

\begin{definition}[Reparameterization]\label{def:reparam}
Let \(\varphi(\overline{x})\) be an $\MSO$ formula and let \(\cC\) be a class of
structures.
A formula \(G(\overline{x},\overline{y})\) is a (functional) \emph{reparameterization} of
\(\varphi\) over \(\cC\) if the following conditions hold:
\begin{description}

  \item[\emph{Same domain:}]
  Over \(\cC\),
$
    \varphi(\overline{x}) \;\equiv\; \exists \overline{y}\, G(\overline{x},\overline{y}).
$
 \item[\emph{Functional:}]
  $G$ defines the graph of a partial function, that is,
  $
    \forall \overline{x}\;\exists^{\le 1}\,\overline{y}\; G(\overline{x},\overline{y}).
$

  \item[\emph{Bounded preimage:}]
  There exists \(N\in\mathbb{N}\) such that, for every \(\mM\in\cC\) and every
  parameter tuple \(\overline{y}\),
  \[
    \bigl|\{\,\overline{x} \mid \mM \models G(\overline{x},\overline{y})\,\}\bigr|
    \;\le\; N.
  \]
\end{description}
We refer to \(\overline{x}\) as the \emph{domain variables} and to
\(\overline{y}\) as the \emph{image variables} of \(G\). The \emph{dimension} of $G$ is the number of variables in $\overline{y}$.
\end{definition}

\noindent
\emph{Intuition.}
A reparameterization replaces the original parameters \(\overline{x}\) by
auxiliary parameters \(\overline{y}\), while preserving the interpreted domain.
The bounded preimage condition ensures that each choice of \(\overline{y}\)
corresponds to only a uniformly bounded number of   \(\overline{x}\).

If \(\varphi(\bar x)\) admits a reparameterization with \(d\) image variables,
then any interpretation whose domain is defined by \(\varphi\) can be reduced
to a \(d\)-dimensional point interpretation
(see \Cref{prop:interpretation-parameter}).

\medskip

Bojańczyk~\cite{Boj23a} studied \emph{polyregular functions},
namely string-to-string functions defined by MSO interpretations.
He proved that it is decidable whether an $\MSO$ formula 
 admits a $d$-dimensional reparameterization over
finite words.
In this work, we generalize this result to arbitrary countable labelled
chains.

The paper is organized as follows.
In \Cref{sect:prel} we recall standard notions concerning chains and
monadic second-order logic, and introduce the basic tools of the
composition method used throughout the paper.
In \Cref{sect:interpretation}  we recall the notion of interpretation and motivate the
reparameterization problem by showing how $d$-dimensional
reparameterizations yield $d$-dimensional interpretations.
\Cref{sect:reparameterization} contains our main result, showing that it is
decidable whether an $\MSO$ formula admits a $d$-dimensional
reparameterization over the class of countable chains.
In \Cref{sect:growth}, following Bojańczyk~\cite{Boj23a}, we study the growth rate of $\MSO$ interpretations.
Extending results of ~\cite{Boj23a,GallotLhoteNguyen:PolyGrowth:2025}, we show that  
  the growth rate is polynomial, with degree equal to the
minimal reparameterization dimension. Due to space limitations, we  provide only proof
sketches. 

\section{Preliminaries} \label{sect:prel}
 
Here  we recall standard notions concerning chains and
monadic second-order logic \cite{Rabin,gurevich1985monadic,thomas1990automata}, and introduce the basic tools of the
composition method \cite{shelah1975monadic,gurevich1985monadic} used throughout the paper.

\subsection{Chains}
A \emph{$k$-chain} is a structure
$
\mathcal M = (M, <, \overline P),
$
where $(M,<)$ is a linear order and $\overline P = (P_1, \dots, P_k)$ is a
$k$-tuple of unary (monadic) predicates on $M$.
When $k$ is clear from the context,  
we simply refer to
$\mathcal M$ as a \emph{chain} or a \emph{labelled chain}.

  An interval $ I$ is a subset  $I\subseteq M$ such that if  $b,c\in I$ and $b<d<c$ then $d\in I$. 
We denote by $\mathcal M\!\upharpoonright I$ the induced substructure on
$I$.

\subsection{Monadic Second-Order Logic}\label{sect:mso}

We work in monadic second-order logic ($\MSO$) over the signature 
$\{<, P_1, \dots, P_k\}$, where each $P_i$ is a unary (monadic) predicate symbol.
First-order variables $x, y, \dots$ range over elements of the domain, and second-order
variables $X, Y, \dots$ range over subsets of the domain.
The atomic formulas are $x < y$, $x = y$, $P_i(x)$, and $X(x)$.
Formulas are built from atomic formulas using Boolean connectives and the quantifiers
$\exists x$, $\forall x$, $\exists X$, and $\forall X$.

\begin{definition}[Definability]\label{def:definability}
Let $\psi(z_1, \dots, z_m)$ be an $\MSO$ formula and let $\mM$ be a structure with domain $D$.
The relation \emph{defined} by $\psi$ in $\mM$ is the set
\[
\psi^{\mM}
:=
\{(a_1, \dots, a_m) \in D^m \mid \mM \models \psi(a_1, \dots, a_m)\}.
\]
A relation $R \subseteq D^m$ is said to be \emph{$\MSO$-definable} in $\mM$ if
there exists an $\MSO$ formula $\psi(z_1, \dots, z_m)$ such that
$R = \psi^{\mM}$.
\end{definition}
We use standard abbreviations. In particular,
$\exists^{\geq N} x\, \varphi(x)$ expresses that there are at least $N$
distinct elements satisfying $\varphi(x)$.

\subsection{Types} 

The \emph{quantifier rank} $\qr(\varphi)$ of an $\MSO$ formula $\varphi$
is the maximal nesting depth of quantifiers.
For $r,k\in\Nat$, let $\Formula{r}{k}$ be the set of $\MSO$ sentences of
quantifier rank at most~$r$ over the signature
$\sigma_k=\{<,P_1,\dots,P_k\}$.

For chains $\frakA$ and $\frakB$, write $\frakA\equiv^r_k\frakB$ if
$\frakA$ and $\frakB$ satisfy the same formulas in $\Formula{r}{k}$.
This is an equivalence relation with finitely many classes.

\begin{lemma}[Hintikka]\label{lem:hintikka-partition}
For every $r,k\in\Nat$ there exists a finite, effectively computable set
$\Hint{r}{k}\subseteq\Formula{r}{k}$ such that:
\begin{enumerate}
  \item The disjunction $\bigvee_{\tau\in\Hint{r}{k}} \tau$ is valid.
  \item If $\tau,\tau'\in\Hint{r}{k}$ are distinct, then
  $\tau\wedge\tau'$ is unsatisfiable.
  \item Moreover, given $\varphi\in \Formula{r}{k}$, one can effectively compute a set
  $H_\varphi \subseteq \Hint{r}{k}$ such that $\varphi$ is equivalent to the disjunction
  of the formulas in $H_\varphi$.
\end{enumerate}
\end{lemma}
By \textup{(1)} and \textup{(2)}, every $\sigma_k$-structure $\frakA$ satisfies a unique
sentence $\tau\in\Hint{r}{k}$. We denote this sentence by $\Typ{r}{k}(\frakA)$ and call it the
\emph{$(r,k)$-type} of~$\frakA$.
Moreover, $\Typ{r}{k}(\frakA)$ effectively determines which sentences in
$\Formula{r}{k}$ are satisfied in~$\frakA$.

When $k$ is clear from the context, we write $\Typ{r}{}(\frakA)$ and call it the
\emph{$r$-type} of $\frakA$. For an interval $I$ of a chain $\frakA$, we write
$\Typ{r}{}(I)$ for the $r$-type of the induced substructure
$\frakA\!\upharpoonright I$.

\subsection{Concatenation of Chains and Sum of Types}
Given $\ell$-chains $\frakA$ and $\frakB$, their \emph{sum} (or
\emph{concatenation})
$\frakA+\frakB$ is obtained by placing $\frakA$ before
$\frakB$.
More generally, for a finite sequence of $\ell$-chains
$\frak A_1,\dots,\frak A_n$, we write
$
\sum_{i=1}^{n} \frak A_i
$
for their sum ~\cite{shelah1975monadic}.
The next proposition implies that the $r$-type of the concatenation of two chains  is determined by their $r$-types.
\begin{proposition}[Sum of types]\label{prop:congr}
For every $r$ there is a computable operation $+$ on $r$-types such that
\[
\Typ{r}{}(\frakA+\frakB)
=
\Typ{r}{}(\frakA)+\Typ{r}{}(\frakB).
\]
\end{proposition}
The following lemma is a standard consequence of composition methods for $\MSO$ over linear orders~\cite{shelah1975monadic}; in particular, it follows from \Cref{prop:congr}.
\begin{lemma}[Local normal form]\label{lem:local}
Let
\[
\varphi(x_1,\dots,x_m)\;\wedge\;\bigwedge_{i<m} x_i<x_{i+1}
\]
be an $\MSO$ formula of quantifier rank~$q$.
For every $r \ge q+2$, the formula $\varphi$ is equivalent to a (finite) positive 
Boolean combination of formulas asserting the $r$-types of the induced
subchains on the intervals
\[
(-\infty,x_1),\quad [x_i,x_{i+1}),\quad \text{and}\quad [x_m,\infty).
\]
Moreover, this Boolean combination is computable from~$\varphi$.
\end{lemma}

\section{Interpretations and Reparameterizations}\label{sect:interpretation}

In this section we recall the notion of interpretation and show that if
the universe formulas admit $d$-dimensional reparameterizations, then
the interpretation is equivalent to a $d$-dimensional interpretation.
These results are not needed for the computability of
$d$-dimensional reparameterizations and serve only to motivate the
reparameterization problem.
Throughout this section we consider $\MSO$ over arbitrary relational
signatures; its syntax and semantics are defined in the natural way.
Readers interested primarily in the decidability results may skim this section on first reading.

\begin{definition}[MSO interpretation]
An \emph{MSO (many-dimensional point) interpretation} is a function
\[
f \colon \mathcal C \to \mathcal D
\]
between two classes of structures, where all structures in $\mathcal C$
(respectively, $\mathcal D$) are over a fixed relational signature.
The interpretation is defined as follows.

All formulas mentioned below are $\MSO$ formulas over the  signature of
the input class $\mathcal C$. All free variables are first-order (element)
variables, but the formulas may quantify over sets.
\begin{enumerate}
\item \textbf{Components.}
There is a finite set $Q$ whose elements are called the
\emph{components} of the interpretation.
Each component $q\in Q$ is associated with a \emph{dimension}
$\dim(q)\in\{0,1,2,\dots\}$.

\item \textbf{Universe formulas.}
For each component $q\in Q$, there is an associated \emph{universe
formula} $\varphi_q(\bar x)$, where the tuple $\bar x$ has length
$\dim(q)$.
These formulas define the universe/domain of the output structure as follows.
For an input structure $\frakA\in\mathcal C$ with universe $A$, the universe of
$f(\frakA)$ is the disjoint union
\[
\coprod_{q\in Q}
\bigl\{
\bar a \in A^{\dim(q)} \mid
\frakA \models \varphi_q(\bar a)
\bigr\}.
\]
\item \textbf{Relation interpretations.}
Let $R$ be a relation symbol in the vocabulary of the output class
$\mathcal D$ of arity $\ell$.
For every choice of components $q_1,\dots,q_\ell\in Q$, there is an
$\MSO$ formula
\[
\varphi^{R}_{q_1,\dots,q_\ell}
(\bar x_1,\dots,\bar x_\ell),
\]
where each tuple $\bar x_i$ has length $\dim(q_i)$.
For every input structure $\frakA\in\mathcal C$ and tuples
$\bar a_i\in A^{\dim(q_i)}$, we have
\[
\frakA \models
\varphi^{R}_{q_1,\dots,q_\ell}(\bar a_1,\dots,\bar a_\ell)
\]
if and only if, in the structure $f(\frakA)$, the relation $R$
holds on the $\ell$-tuple whose $i$-th coordinate is $\bar a_i$ from
component $q_i$.
\end{enumerate}
\end{definition}
\begin{remark}\label{rem:int-injective}
The interpretations defined above are usually called
\emph{injective interpretations}:
every element of the output structure
$\frakB := f(\frakA)$ is represented by a unique tuple of elements of
$\frakA$.

In more general notions of interpretation, one allows a binary relation
$E$ on the domain of $f(\frakA)$, defined by $\MSO$ formulas over
$\frakA$.
For every structure $\frakB := f(\frakA)$, the relation $E$ is required
to be a \emph{congruence}, that is, an equivalence relation such that all
relations of $\frakB$ are invariant under~$E$.
In this case, the elements of the output structure correspond to the
$E$-equivalence classes.
\end{remark}
An interpretation is said to be \emph{$m$-dimensional} if all its universe
 formulas use at most $m$ free variables.
Two interpretations
\[
f,h \colon \mathcal C \to \mathcal D
\]
are \emph{equivalent} if for every $\frakA\in\mathcal C$ the structures
$f(\frakA)$ and $h(\frakA)$ are isomorphic.

\begin{proposition}\label{prop:interpretation-parameter}
Let $\mathcal C$ be a class of structures equipped with an
$\MSO$-definable linear order.
Assume that for each component $q$ of an interpretation $f$ over
$\mathcal C$, the corresponding universe formula $\varphi_q$ admits a
reparameterization of dimension $d$.
Then $f$ is equivalent to a $d$-dimensional interpretation $h$.
Moreover, there exists an $\MSO$ formula that uniformly defines, in
every $\frakA\in\mathcal C$, an isomorphism between
$f(\frakA)$ and $h(\frakA)$.
\end{proposition}
\begin{proof} 
 For simplicity, we present the proof only for injective interpretations
(see \Cref{rem:int-injective}).
The argument easily extends to the general case, including
interpretations that use parameters~\cite{Friedman1I}.
 
Let $f\colon\mathcal C\to\mathcal D$ be the given MSO interpretation.
For each component $q$ of $f$, let
\[
G_q(\bar x,\bar y)
\]
be an $\MSO$-definable reparameterization of the corresponding domain
formula $\varphi_q(\bar x)$, where $|\bar y|\le d$.
Assume that there exists a constant $N_q$ such that for every input
structure $\frakA\in\mathcal C$ and every tuple $\bar y$, the set
\[
\{\bar x \mid \frakA\models  G_q(\bar x,\bar y)\}
\]
has size at most $N_q$.

Note that
\[
G_q(\bar x,\bar y)\rightarrow \varphi_q(\bar x)
\]
follows from the domain clause of the definition of reparameterization.
Hence,
\[
\{\bar x \mid \frakA\models  \varphi_q(\bar x)\wedge G_q(\bar x,\bar y)\}
=
\{\bar x \mid \frakA\models  G_q(\bar x,\bar y)\}.
\]
We construct an equivalent $d$-dimensional MSO interpretation
$h\colon\mathcal C\to\mathcal D$ as follows.

\smallskip
\noindent\textbf{Components and domain formulas.}
The components of $h$ are pairs $(q,i)$, where $q$ is a component of $f$
and $1\le i\le N_q$.
For each such component $(q,i)$, the domain formula
$\psi_{q,i}(\bar y)$ is defined by
\[
\psi_{q,i}(\bar y)
\;:=\;
\exists \bar x\;
\Bigl(
G_q(\bar x,\bar y)
\;\wedge\;
\bar x \text{ is the $i$-th tuple in the lexicographic order of this set}
\Bigr).
\]
Since lexicographic order on tuples is $\MSO$-definable and $|\bar y|\le
d$, each $\psi_{q,i}$ is an $\MSO$ formula with at most $d$ free
variables.

\smallskip
\noindent\textbf{Relation formulas.}
Let $R$ be a relation symbol of the output vocabulary of arity $\ell$.
For every tuple of components $(q_1,i_1),\dots,(q_\ell,i_\ell)$, we define
the corresponding relation formula
\[
\psi^{R}_{(q_1,i_1),\dots,(q_\ell,i_\ell)}(\bar y_1,\dots,\bar y_\ell)
\]
to assert that there exist tuples $\bar x_1,\dots,\bar x_\ell$ such that
for each $j\le\ell$,
$\bar x_j$ is the $i_j$-th tuple satisfying
$  G_{q_j}(\bar x_j,\bar y_j)$, and
\[
\varphi^{R}_{q_1,\dots,q_\ell}(\bar x_1,\dots,\bar x_\ell)
\]
holds.
By the boundedness assumption on the reparameterizations, for each
$\bar y_j$ there is at most one such tuple $\bar x_j$, so the relation is
well defined.

\smallskip
\noindent\textbf{Equivalence of interpretations.}
For every input structure $\frakA\in\mathcal C$, define a mapping
\[
\pi\colon h(\frakA)\to f(\frakA)
\]
by sending an element $\bar y$ of component $(q,i)$ to the unique tuple
$\bar x$ of component $q$ such that
\[
\frakA\models  G_q(\bar x,\bar y)
\]
and $\bar x$ is the $i$-th tuple in the corresponding lexicographic
ordering.
The boundedness assumption ensures that $\pi$ is well defined and
bijective.

By construction, $\pi$ preserves all relations, since a relation holds
between tuples of $\bar y$-elements in $h(\frakA)$ if and only if the
corresponding tuples of $\bar x$-elements satisfy the defining relation
formula in $f(\frakA)$.
It is easy to verify that $\pi$ is $\MSO$-definable in $\frakA$.
Hence $\pi$ is an $\MSO$-definable isomorphism between $h(\frakA)$ and
$f(\frakA)$.

Therefore, the interpretation $h$ is $d$-dimensional and is equivalent to 
the original interpretation~$f$.\qed
\end{proof}

\section{Reparameterization}\label{sect:reparametrization} \label{sect:reparameterization} 

\noindent
The following theorem generalizes Lemma~II.7 of Bojańczyk~\cite{Boj23a}
from finite words to arbitrary countable chains.

\begin{theorem}[Decidability   of $m$-dimensional reparameterizations]
\label{thm:m-dim-reparam}
Let $\cC$ be an $\MSO$-definable class of countable chains, and let
$\varphi(\bar x)$ be an $\MSO$ formula with first-order free variables
$\bar x=(x_1,\dots,x_n)$.
There is an algorithm which, given $m\in\mathbb{N}$, decides whether there exists
an $\MSO$ formula $G(\bar x,\bar y)$ with $\bar y$ an $m$-tuple of first-order variables,
and a number $N\in\mathbb{N}$ such that, over $\cC$,
\[
  \varphi(\bar x)\ \equiv\ \exists \bar y\, G(\bar x,\bar y)
  \qquad\text{and}\qquad
  \bigl|\{\,\bar x \mid \mM\models G(\bar x,\bar y)\,\}\bigr|\ \le\ N
\]
for all $\mM\in\cC$ and all tuples $\bar y$ of elements of $\mM$.
Moreover, whenever the answer is positive,   an $m$-dimensional functional
reparameterization can be effectively constructed from $\varphi$ and an $\MSO$
definition of $\cC$.
\end{theorem}
Below we consider the reparameterization problem over the class of
countable chains.
If $\mathcal{C}$ is a class of countable chains definable by a formula
$\psi$, then the reparameterization problem for $\varphi$ over
$\mathcal{C}$ is reducible to the reparameterization problem for
$\varphi \wedge \psi$ over the class of all countable chains.
\subsection{Pumpability and Dimension Bounds}

In this subsection we study the connection between pumpability and the minimal dimension of a reparameterization.
We first establish in \Cref{lem:pump-eq} a pumping criterion for formulas with one free variable, characterizing the existence of arbitrarily many realizations in terms of idempotent types.
We then introduce in \Cref{lem:decrease-dimension} the notion of a pumpable pair of types and show how non-pumpable pairs allow the elimination of variables, yielding lower-dimensional reparameterizations.
Finally, \Cref{lem:no-decrement} shows that if all adjacent pairs are pumpable, then the dimension cannot be reduced.
\begin{lemma}\label{lem:pump-eq}
There exist constants $B,q\in \nat$, computable from  $\varphi(x)$, such that the following are equivalent:
\begin{enumerate}
  \item The formula $\exists^{\geq B}x\,\varphi(x)$ is satisfiable.
  \item For every $N$, the formula $\exists^{\geq N}x\,\varphi(x)$ is satisfiable.
  \item There exist satisfiable  $q$-types $\tau$, $\tau_1$, and $\tau_2$ such that:
  \begin{itemize}
    \item $\tau$ is idempotent (i.e., $\tau + \tau = \tau$),
    \item $\tau_1 + \tau = \tau_1$,
    \item $\tau + \tau_2 = \tau_2$, and
    \item the conjunction $\type^q((-\infty, x)) = \tau_1 \,\wedge\, \type^q([x, \infty)) = \tau_2$ implies $\varphi(x)$.
  \end{itemize}
\end{enumerate}
\end{lemma}

\begin{proof} 
The implication $(2) \implies (1)$ is trivial.
 
$(1) \implies (3)$.
Let $q$ be any number that exceeds the quantifier rank of~$\varphi$ by~$2$.
By \Cref{lem:hintikka-partition} and \Cref{lem:local},
the formula $\varphi(x)$ is equivalent to a disjunction of formulas of the form  
\[
\type^q((-\infty, x)) = \tau_1 \,\wedge\, \type^q([x, \infty)) = \tau_2.
\]
Without loss of generality, we may assume that $\varphi$ is of this form.

We aim to find a constant $B$ for formulas of the above form.

Let $H$ be the number of $q$-types (over the signature of $\vp$).  
By Ramsey's theorem, there exists a constant $B$ such that every coloring of the edges of the complete graph on $B$ vertices using $H$ colors contains a monochromatic triangle.  
We claim this value of $B$ suffices.

Suppose $a_1 < \dots < a_B$ are elements in some chain $\mathcal{M}  $  such that $\mathcal{M} \models \varphi(a_i)$ for all $i$.  
For each pair $i < j$, define the color $\mathsf{col}(i, j)$ as the $q$-type of the interval $[a_i, a_j)$.  
By Ramsey's theorem, there exist indices $i < j < k$ such that 
\[
\mathsf{col}(i, j) = \mathsf{col}(j, k) = \mathsf{col}(i, k) = \tau.
\]
It follows that $\tau + \tau = \tau$, i.e., $\tau$ is idempotent.

By assumption, the type $\tau_1$ of $(-\infty, a_i)$ equals that of $(-\infty, a_j)$, so $\tau_1 + \tau = \tau_1$.  
Similarly, the type $\tau_2$ of $[a_k, \infty)$ equals that of $[a_j, \infty)$, so $\tau + \tau_2 = \tau_2$.  
Thus, condition~\textup{(3)} of the lemma holds. Hence, we have established   $(1) \implies (3)$. 

Let us prove that \textup{(3)} implies \textup{(2)}.
Let $L$, $I$, and $U$ be chains with $q$-types $\tau_1$, $\tau$, and $\tau_2$, respectively.
Since $q\ge 2$, every $q$-type determines whether the corresponding structures have a minimal element.
In particular, all structures of type~$\tau_2$ have a minimal element.
Since $\tau+\tau_2=\tau_2$, it follows that all structures of type~$\tau$ also have a minimal element.
Hence, $I$ has a minimal element.

For $N\in \Nat$, define the structure
\[
\mM_{N} \ :=\  L\ +\ \underbrace{I\ +\ \cdots\ +\ I}_{N\ \text{copies}}\ +\ U.
\]
Then the minimal element of each copy of $I$ satisfies $\varphi(x)$ in~$\mM_N$.
Thus, $\exists^{\geq N}x\,\varphi(x)$ holds in~$\mM_N$ for all $N$. \qed
\end{proof}
\begin{definition}[Pumpable pair of types]
We say that a pair of $q$-types $(\tau_b, \tau_e)$ is \emph{pumpable} if there exists a satisfiable idempotent $q$-type $\tau$ such that:
\[
\tau_b + \tau = \tau_b \quad \text{and} \quad \tau + \tau_e = \tau_e.
\]
\end{definition}

\medskip
\noindent
Let $\tau_0, \dots, \tau_k$ be a sequence of $q$-types, and define the formula

\begin{equation}\label{eq-vp}
\begin{aligned}
\varphi(x_1,\dots,x_k) :=\;&
x_1 < x_2 < \cdots < x_k \,\wedge \\
&
\type^q((-\infty,x_1)) = \tau_0 \,\wedge\,
\type^q([x_k,\infty)) = \tau_k \,\wedge \\
&
\bigwedge_{i=1}^{k-1}
\type^q([x_i,x_{i+1})) = \tau_i .
\end{aligned}
\end{equation}
If the pair $(\tau_{i-1}, \tau_i)$ is not pumpable, then the variable $x_i$ can be eliminated as explained in the following Lemma.

\begin{lemma}\label{lem:decrease-dimension}
Assume that $\varphi$ is as in \Cref{eq-vp}, and that the pair $(\tau_i, \tau_{i+1})$ is not pumpable. Then:
\[
G_i(x_1, \dots, x_k, y_1, \dots, y_{k-1}) := 
\varphi(x_1, \dots, x_k) \,\wedge\, 
\bigwedge_{j < i} y_j = x_j \,\wedge\, 
\bigwedge_{j > i} y_j = x_{j-1}
\]
defines a $(k{-}1)$-dimensional functional reparameterization of $\varphi$.
\end{lemma}

\begin{proof}
Follows from \Cref{lem:pump-eq}. \qed
\end{proof}

\begin{remark}
In the reparameterization constructed in
\Cref{lem:decrease-dimension}, each image variable is equal to one of
the domain variables.
\end{remark}
The next lemma treats the case in which all pairs are pumpable, and
shows that in this situation the dimension cannot be reduced.
\begin{lemma} \label{lem:no-decrement}
Assume that $\varphi(x_1, \dots, x_k)$ is as in \Cref{eq-vp}, and that all pairs $(\tau_i, \tau_{i+1})$ are pumpable. Then:
\begin{enumerate}
  \item For every $N \in \mathbb{N}$ and every proper subset $I \subsetneq \{1, \dots, k\}$,  
  there exists a structure $\mathcal{M}_N$ and at least $N$ distinct  tuples that satisfy $\varphi$ and agree on all coordinates $x_i$  for $i\in I$.
  
  \item If $G$ is a reparameterization of $\varphi$, then $G$ must use at least $k$ image variables.
  \item
 For every $N \in \mathbb{N}$, there exists a structure $\mathcal{M}_N$ and a subset $S \subseteq \mathcal{M}_N$ of size $2Nk$ such that:
 \[
\left\{ (a_1, \dots, a_k) \in S^k \;\middle|\; \mathcal{M}_N \models \varphi(a_1, \dots, a_k) \right\}
\]
has cardinality at least $(2N)^k$.
\end{enumerate}
\end{lemma}
\begin{proof}
 \textup{(1)} follows from \textup{(2)}.  
  
  \noindent \textup{(2)}
Assume toward a contradiction that there exists a reparameterization $G$
of $\varphi$ with $l<k$ image variables.
Let $r-2$ be greater than the quantifier rank  of both $\varphi$ and $G$.
By \Cref{lem:local}, the formulas $\varphi$ and $G$ are equivalent to
Boolean combinations of formulas asserting the $r$-types of the induced
subchains on the relevant intervals.

Let $\tau'_i$ be an idempotent $q$-type witnessing that the pair
$(\tau_i,\tau_{i+1})$ occurring in the representation of $\varphi$ as in
\Cref{eq-vp} is pumpable, that is,
\[
\tau_i + \tau'_i = \tau_i
\quad\text{and}\quad
\tau'_i + \tau_{i+1} = \tau_{i+1}.
\]
There is a satisfiable idempotent $r$-type $\sigma_i$ such that $\sigma_i\rightarrow\tau'_i$.
Let $U_i$ be a chain of $r$-type $\sigma_i$ (and hence of $q$-type $\tau'_i$), and let $L_i$ be a chain of type
$\tau_i$.
Define the structure $\mathcal M := \mathcal M_N$ by
\[
\mathcal M := L_0 + \sum_i \bigl(U_i \times 2N + L_{i+1}\bigr),
\]
where $U_i \times 2N$ denotes the concatenation of $2N$ consecutive copies
of $U_i$.

Let $a_i$ be the first element of the $N$-th copy of $U_i$.
By construction,
\[
\mathcal M \models \varphi(a_1,\dots,a_k).
\]
Hence, there exist elements $b_1,\dots,b_\ell$ such that
\[
\mathcal M \models G(a_1,\dots,a_k,b_1,\dots,b_\ell).
\]
We say that $a_i$ and $b_j$ are \emph{close} if $b_j$ lies in
$U_i \times 2N$; otherwise, we say that $a_i$ and $b_j$ are \emph{far}.
Since $\ell < k$, by the pigeonhole principle there exists some $a_i$
that is far from every $b_j$.

For $m=2,\dots,2N$, let $a_i^m$ be the first element of the $m$-th copy of
$U_i$.
Since the $r$-type  $\sigma_i:=\type^r(U_i)$ is idempotent, for every formula $\psi$ of
quantifier rank  at most $r-2$, and every tuple $\bar c$ that is far from
$a_i$, we have
\[
\mathcal M \models \psi(\bar c,a_i)
\quad\text{iff}\quad
\mathcal M \models \psi(\bar c,a_i^m).
\]
It follows that all tuples
\[
(a_1,\dots,a_i^m,\dots,a_k),
\qquad m=2,\dots,2N,
\]
are mapped by $G$ to the same tuple $(b_1,\dots,b_\ell)$.
Hence, the preimage of $(b_1,\dots,b_\ell)$ under $G$ has size at least
$2N-1$.

This contradicts the bounded preimage condition in the definition of
reparameterization.
Therefore, $G$ cannot be a valid reparameterization of $\varphi$.

\noindent \textup{(3)}
 Let $\mathcal{M} = \mathcal{M}_N$ be the structure defined as in the proof of  (2), where for each $i = 1, \dots, k$, the structure contains $2N$ consecutive copies of a  chain  $U_i$ of type $\sigma_i$ and hence of $\tau'_i$.

For each $i = 1, \dots, k$, define the set
\[
S_i := \left\{ l_i^j \;\middle|\; \text{$l_i^j$ is the first element in the $j$-th copy of $U_i$} \right\}.
\]
Note that $|S_i| = 2N$. Define $S := \bigcup_{i=1}^k S_i$. Then $|S| \leq 2Nk$.

Now consider any tuple $(a_1, \dots, a_k)$ such that $a_i \in S_i$ for all $i$. By construction, each $a_i$ lies at the beginning  of a block of type $U_i$ positioned to satisfy the formula $\varphi$. Therefore, $ 
\mathcal{M} \models \varphi(a_1, \dots, a_k).
$

Since there are $2N$ choices for each $a_i \in S_i$, the total number of such tuples is $(2N)^k$. This shows that
 $
\left| \left\{ (a_1, \dots, a_k) \in S^k \;\middle|\; \mathcal{M} \models \varphi(a_1, \dots, a_k) \right\} \right| \geq (2N)^k,
$
as required. \qed
\end{proof}
\subsection{Computability of a minimal reparameterization}
\begin{definition}[Minimal reparameterization]
A reparameterization $G(\bar x,\bar y)$ of a formula $\varphi(\bar x)$ is called
\emph{minimal} if 
it uses the smallest possible number
of image variables among all reparameterizations of $\varphi$.
\end{definition}
\begin{proposition}\label{prop:same-min}

Let $G(\overline{x}, \overline{y})$ be a reparameterization of a formula $\varphi(\overline{x})$. Then:
 There exists an algorithm which, given a $d$-dimensional reparameterization of $\varphi$ (respectively, of $\exists \overline{x}\, G$), constructs a $d$-dimensional reparameterization of $\exists \overline{x}\, G$ (respectively, of $\varphi$).

\end{proposition}

\begin{proof}


   Assume that $F(\overline{x}, \overline{u})$ is a $d$-dimensional reparameterization of $\varphi(\overline{x})$.  
  Then we can define a reparameterization of $\exists \overline{x}\, G(\overline{x}, \overline{y})$ by:
  \[
  \exists \overline{x} \left( 
  F(\overline{x}, \overline{u}) \;\wedge\; 
  \text{$\overline{x}$ is the lexicographically minimal tuple such that } G(\overline{x}, \overline{y})
  \right).
  \]
 That is, for each tuple $\overline{y}$, we select the lexicographically minimal tuple $\overline{x}$ such that $G(\overline{x}, \overline{y})$ holds, and apply the original reparameterization $F$ to that tuple. 

Note that, by the definition of reparameterization, there are only finitely many tuples $\overline{x}$ such that $G(\overline{x}, \overline{y})$ holds. Therefore, the lexicographically minimal such $\overline{x}$ exists.

This construction yields a $d$-dimensional reparameterization of $\exists \overline{x}\, G$.

  \smallskip

  Conversely, assume that $H(\overline{y}, \overline{u})$ is a $d$-dimensional reparameterization of the formula $\exists \overline{x}\, G(\overline{x}, \overline{y})$.  
  Then we can define a reparameterization of $\varphi(\overline{x})$ by:
  \[
  \exists \overline{y} \left( G(\overline{x}, \overline{y}) \,\wedge\, H(\overline{y}, \overline{u}) \right).
  \]
  That is, for each $\overline{x}$ satisfying $\varphi$, we compute the image of $\overline{x}$ by first obtaining $\overline{y}$ via $G$, and then applying the reparameterization $H$ to $\overline{y}$. This defines a $d$-dimensional reparameterization of $\varphi$. \qed
\end{proof}
\begin{lemma}\label{lem:disj-max}
Let $\psi_1(\bar x),\dots,\psi_k(\bar x)$ be formulas.
Assume that for each $j $ the formula $\psi_j(\bar x)$ admits a
$d_j$-dimensional reparameterization.
Then one can effectively construct a $d$-dimensional reparameterization of
$\bigvee_{j=1}^k \psi_j(\bar x)$, where $d:=\max_{1\le j\le k} d_j$.
\end{lemma}

\begin{proof}
Assume $G_j(\overline{x}, \overline{y})$ is   a $d_j$-dimensional reparameterization  of $\psi_j$.
Let  $d = \max_j d_j$.  
Define the   disjunction
\[
G(\overline{x}, \overline{y}) := \bigvee_j \left( \left( \bigwedge_{i<j} \neg \psi_i(\overline{x}) \right) \wedge \psi_j(\overline{x}) \wedge G_j(\overline{x}, \overline{y}) \right).
\]
This formula selects the appropriate $G_j$ for the first $\psi_j$ (in index order) that is satisfied by $\overline{x}$.  
By construction, $G$ is a $d$-dimensional reparameterization of $\bigvee_j \psi_j$.  \qed
\end{proof}
\Cref{thm:m-dim-reparam} follows from the next proposition.
\begin{proposition}
  A minimal reparameterization of an $\MSO$ formula is computable.
\end{proposition}

\begin{proof}
We proceed by induction on the number $k$ of free variables in the formula.

\medskip
\noindent\textbf{Base case:} $k = 0$.\\
In this case, $\varphi$ is a sentence with no free variables. The minimal dimension of a reparameterization is $0$, and the trivial reparameterization is clearly computable.

\medskip
\noindent\textbf{Inductive step:} Assume that for all formulas with fewer than $k$ free variables, a minimal reparameterization is computable.

Let $\varphi$ be a formula with $k$ free variables. We can effectively compute a disjunction $\bigvee_i \varphi_i$ equivalent to $\varphi$, where each $\varphi_i$ is of the form described in \Cref{eq-vp}.

Suppose first that for every disjunct $\varphi_i$, the corresponding sequence of types $(\tau^i_0, \dots, \tau^i_k)$ contains at least one non-pumpable pair. Then, by \Cref{lem:decrease-dimension}, each such $\varphi_i$ admits a $(k{-}1)$-dimensional reparameterization $G_i$, computable from $\varphi_i$.
By taking a   disjunction of these reparameterizations (as in the proof of
\Cref{lem:disj-max}), we obtain a $(k{-}1)$-dimensional reparameterization
$G(\overline{x},\overline{y})$ of $\varphi$, computable from $\varphi$.

By the inductive hypothesis, we can compute a minimal reparameterization of the formula $\exists \overline{x}\, G(\overline{x}, \overline{y})$. Then, by \Cref{prop:same-min},  we can compute  a minimal reparameterization of $\varphi$.

Now suppose there exists a disjunct $\varphi_i$ whose sequence of types $(\tau^i_0, \dots, \tau^i_k)$ contains only pumpable pairs. Then, by \Cref{lem:no-decrement}, every reparameterization $G_i$ of $\varphi_i$ has dimension at least $k$.

Note that if $G$ is a reparameterization of $\bigvee_i \varphi_i$, then $G \wedge \varphi_i$ is a reparameterization of $\varphi_i$. Hence, any reparameterization of $\varphi$ must have dimension at least $k$.

It follows that the trivial reparameterization
$
G(\overline{x}, \overline{y}) := \varphi(\overline{x}) \,\wedge\, \bigwedge_j y_j = x_j
$
is minimal and computable. \qed
\end{proof}
\begin{remark}
\label{rem:canonical-reparameterization}
By the construction above, there exists a   reparameterization
\(G(\overline{x},\overline{y})\) of \(\varphi(\overline{x})\) over \(\cC\) of
\emph{minimal dimension} (i.e., with \(|\overline{y}|\) minimal) such that each
image variable  \(y_j\) is always equal to one of the domain  variables. 
\end{remark}

\section{Growth Rate}\label{sect:growth}
Bojańczyk  \cite{Boj23a} introduced the notion of the \emph{growth rate} of a
string-to-string function as the function that maps an input length
$n\in\{0,1,\dots\}$ to the maximal size of an output produced on inputs of
length at most~$n$.

This notion was later generalized by Gallot et~al. \cite{GallotLhoteNguyen:PolyGrowth:2025}, who showed that the
output size of an $\MSO$ set interpretation from finite trees to relational
structures grows either polynomially or exponentially in the input size,
with a computable degree $k\in\Nat\cup\{\infty\}$.

In this section, we define the growth rate of an $\MSO$ interpretation over
arbitrary structures, including infinite ones.
We show that the growth rate of an $\MSO$ \emph{point} interpretation is
polynomial, with degree equal to the minimal reparameterization
dimension.

Let $\varphi(x_1, \dots, x_k)$ be an $\MSO$ formula, and let $S$ be a subset of a structure $\mathcal{M}$. Define:
$ 
I(S, \mathcal{M}) := \left\{ \overline{a} \in S^k \mid \mathcal{M} \models \varphi(\overline{a}) \right\}.
$
Let the growth function $g_\varphi(n, \mathcal{M})$ be defined as:
\[
g_\varphi(n, \mathcal{M}) := \max \left\{ \left| I(S, \mathcal{M}) \right| \;\middle|\; S \subseteq M,~ |S| \leq n \right\}.
\]
Note that $g_\varphi(n, \mathcal{M})$ is bounded above by $n^k$.    

For a class $\mathcal{C}$ of structures, define:
\[
g_\varphi(n, \mathcal{C}) := \max_{\mathcal{M} \in \mathcal{C}} g_\varphi(n, \mathcal{M})
\]
The following theorem generalizes Bojańczyk's theorem~\cite[Theorem~II.3]{Boj23a}
 from finite words to arbitrary countable chains.

\begin{theorem}\label{th:growth}
Let $\varphi(\overline{x}) \in \MSO$ and let $\mathcal{C}$ be an $\MSO$-definable class of labelled linear orders. Then there exists a natural number $d \in \mathbb{N}$ such that
\[
g_\varphi(n, \mathcal{C}) = \Theta(n^d).
\]
Furthermore, the degree \(d\) coincides with the minimal
reparameterization dimension of \(\varphi\) over \(\mathcal C\), and is
therefore computable.
\end{theorem}
\begin{proof} 

Let $\varphi(\overline{x})$ be an $\MSO$ formula with $k$ free variables, and let $\mathcal{C}$ be an $\MSO$-definable class of labelled linear orders. 
By \Cref{thm:m-dim-reparam}, for any such $\varphi$, we can compute a
minimal $d$-dimensional reparameterization
\[
G(\overline{x}, \overline{y})
\]
of $\varphi$ over the class $\mathcal{C}$.
Moreover, we can assume that $G$ is functional.

\noindent
Note that for a class $\mathcal{C}$ of countable labelled chains definable by a sentence $\psi$, the growth function $g_\varphi(n, \mathcal{C})$ coincides with $g_{\varphi \wedge \psi}(n, \mathit{Countable})$, where $\mathit{Countable}$ denotes the class of all countable labelled chains.

\medskip
\noindent
From now on, we consider only $g_{\varphi}(n) := g_{\varphi}(n, \mathit{Countable})$, the growth function over the class of all countable labelled chains.

\medskip
\noindent\textbf{Upper Bound: $g_\varphi(n) = O(n^d)$} \\
There exists a   minimal reparameterization  $G(\overline{x}, \overline{y})$ which implies  $\bigwedge_i \bigvee_j (y_i = x_j)$.
\begin{equation}\label{eq-prop-repar}
  \mathcal{M} \models G(\overline{a}, \overline{b}) \rightarrow \bigwedge_i \bigvee_j (b_i = a_j).
\end{equation}
\noindent
The growth function $g_\varphi(n)$ is defined as the maximum size of the set
\[
\left\{ \overline{a} \in S^k \mid \mathcal{M} \models \varphi(\overline{a}) \right\}
\]
where $S \subseteq M$ with $|S| \leq n$, and $\mathcal{M}$ ranges over all countable labelled chains.
Due to the reparameterization and \Cref{eq-prop-repar}, this set corresponds to the projection of
\[
\left\{ (\overline{a}, \overline{b}) \in S^{k+d} \mid \mathcal{M} \models G(\overline{a}, \overline{b}) \right\}.
\]
Therefore, the number of satisfying $k$-tuples is bounded by the number of satisfying $(k+d)$-tuples, which is at most $B \cdot n^d$, because for each $\overline{b} \in S^d$, the number of $\overline{a}$ such that $G(\overline{a}, \overline{b})$ holds is bounded by a constant $B$. This establishes the upper bound $O(n^d)$.

\medskip
\input{th-growth-app}

\noindent\textbf{Computability of $d$:} \\
Since the minimal reparameterization of $\varphi$ is computable (by prior results), and its dimension $d$ is minimal by definition, we conclude that $d$ is computable from $\varphi$ and a definition of the class $\mathcal{C}$.
\qed
\end{proof}

\section{Conclusion and Further Results}

Bojańczyk~\cite{Boj23a} studied  reparameterizations of $\MSO$-definable relations
over finite words and proved that it is decidable whether an
$m$-dimensional point $\MSO$ interpretation admits a $d$-dimensional point
$\MSO$ reparameterization.
In this paper, we extend this decidability result to arbitrary countable
labelled chains.

In many situations, however, one considers MSO interpretations that are more general than point interpretations.
In \emph{finite-set} and \emph{set} interpretations, elements of the interpreted
structure are represented by tuples of finite sets or arbitrary sets,
respectively.

\paragraph{Interpretations of \(\mathfrak A\) in \(\mathfrak B\).}
We distinguish three standard flavours of $d$-dimensional $\MSO$ interpretations
of a structure \(\mathfrak A\) in a structure \(\mathfrak B\).
In each case, an element of \(\mathfrak A\) is encoded inside \(\mathfrak B\) by a
\(d\)-tuple:
\begin{description}
  \item[\textbf{Set interpretations} (\(\MSO\), $d$-dimensional):]
  each element of \(\mathfrak A\) is represented by a \(d\)-tuple of arbitrary
  subsets of \(\mathfrak B\).
  \item[\textbf{Finite-set interpretations} (\(\MSO\), $d$-dimensional):]
  each element of \(\mathfrak A\) is represented by a \(d\)-tuple of finite
  subsets of \(\mathfrak B\).
  \item[\textbf{Point interpretations} (\(\MSO\), $d$-dimensional):]
  each element of \(\mathfrak A\) is represented by a \(d\)-tuple of elements of
  \(\mathfrak B\).
\end{description}
A natural question is whether a given interpretation can be simplified; for
instance, whether a set or finite-set interpretation is equivalent to a point
interpretation.

Gallot, Lhote, and Nguyen~\cite{GallotLhoteNguyen:PolyGrowth:2025} studied
\emph{finite-set-to-point} reparameterizations over finite words and trees and proved
that it is decidable whether a finite-set interpretation admits a point
reparameterization.

In~\cite{Rabinovich:ICALP26} we extended this decidability result to arbitrary countable labelled chains:
given an $\MSO$ formula \(\varphi\) with free finite-set variables, it is decidable
whether \(\varphi\) admits a point reparameterization.

Over linear orders, a canonical class of simple subsets is given by
\emph{cuts}, i.e., downward-closed sets.
Cuts correspond to points in the Dedekind completion and form a
structurally robust family of subsets.
The analogue of our main result, \Cref{thm:m-dim-reparam}, also holds when
first-order domain and image variables are replaced by variables
interpreted as cuts.

Finally, building on the techniques of~\cite{BKR11}, we also proved
in~\cite{Rabinovich:ICALP26} that it is decidable whether an $\MSO$
formula \(\varphi\) with free set variables admits a
many-dimensional finite-set reparameterization over Dedekind-complete
countable labelled chains.

\bibliographystyle{plain}
\bibliography{sample} 

\end{document}

%% file: th-growth-app.tex
\medskip
\noindent\textbf{Lower Bound: \(g_\varphi(n)=\Omega(n^d)\).}

Let 
$
\exists \bar x\,G(\bar x,\bar y)
$ be a minimal functional reparameterization of $\vp$.

It is equivalent to a finite disjunction of formulas \(D_j(\bar y)\),
each of the form
\begin{equation}\label{eq-vp4}
\type^q((-\infty,y_1))=\tau_0
\;\wedge\;
\bigwedge_{i=1}^{d-1}
\type^q([y_i,y_{i+1}))=\tau_i
\;\wedge\;
\type^q([y_d,\infty))=\tau_d .
\end{equation}
for suitable \(q\)-types \(\tau_0,\dots,\tau_d\).

We claim that for some disjunct \(D_j\), every adjacent pair
\((\tau_i,\tau_{i+1})\) is pumpable. Indeed, if every disjunct contained
a non-pumpable adjacent pair, then by
\Cref{lem:decrease-dimension} each \(D_j\) would admit a
reparameterization of dimension strictly smaller than \(d\). Hence
\(\exists \bar x\,G(\bar x,\bar y)\) would also admit a
reparameterization of dimension smaller than \(d\), and by
\Cref{prop:same-min} the same would hold for \(\varphi\), contradicting
the minimality of \(d\).

Fix such a disjunct \(D:=D_j\). We construct, for every \(n\), a
countable labelled chain \(\mathcal M_n\) and a set \(S\subseteq M_n\)
of size \(O(n)\) such that \(\varphi\) has \(\Omega(n^d)\) satisfying
tuples in \(S^k\).

For each \(i\), let \(\tau'_i\) be an idempotent \(q\)-type witnessing
that \((\tau_i,\tau_{i+1})\) is pumpable, that is,
\[
\tau_i+\tau'_i=\tau_i
\qquad\text{and}\qquad
\tau'_i+\tau_{i+1}=\tau_{i+1}.
\]
Let \(U_i\) be a labelled chain of type \(\tau'_i\), and let \(L_i\) be
a labelled chain of type \(\tau_i\). Note that \(U_i\) has a minimal
element, which we choose as the position of \(y_i\).

Set \(m:=2k+3\) and let \(r:=(m+1)/2\). Define
\[
\mathcal M_1 :=
L_0+\sum_{i=1}^{d}
\bigl(U_i\times m+L_i\bigr),
\]
where \(U_i\times m\) denotes the concatenation of \(m\) copies of
\(U_i\).

Let \(b_i\) be the minimal element of the \(r\)-th copy of \(U_i\).
By the choice of the types and by idempotency,
\[
\mathcal M_1\models D(b_1,\dots,b_d).
\]
Since \(D(\bar y)\) implies \(\exists \bar x\,G(\bar x,\bar y)\), choose
\(\bar a=(a_1,\dots,a_k)\) such that
\[
\mathcal M_1\models G(\bar a,\bar b).
\]
For each \(i\), among the \(r-1=k+1\) copies of \(U_i\) to the left of
the \(r\)-th copy, choose one copy containing none of the elements
\(a_1,\dots,a_k\). Similarly, among the \(m-r=k+1\) copies of \(U_i\) to
the right of the \(r\)-th copy, choose one copy containing none of these
elements. Such copies exist because there are only \(k\) elements
\(a_1,\dots,a_k\). We use these two empty copies as buffers.

For each \(i\), define \(C_i\subseteq U_i\) to be the finite set of
positions in \(U_i\) occupied by one of the elements \(a_1,\dots,a_k\)
inside some copy of \(U_i\) in \(\mathcal M_1\):
\[
\begin{aligned}
C_i := \{\, c \mid\;&
\text{for some copy of } U_i \text{ in } \mathcal M_1,\\
&
\text{one of } a_1,\dots,a_k
\text{ occurs at position } c \text{ in that copy}
\,\}.
\end{aligned}
\]
Also set
\[
S_i:=L_i\cap\{a_1,\dots,a_k\}.
\]
For \(n\ge1\), define
\[
\mathcal M_n :=
L_0+\sum_{i=1}^{d}
\bigl(U_i\times(m+n)+L_i\bigr).
\]
Let
\[
S :=
\bigcup_i S_i
\;\cup\;
\bigcup_i \bigl(C_i\times\{1,\dots,m+n\}\bigr).
\]
Since each \(C_i\) and each \(S_i\) has bounded size independently of
\(n\), we have \(|S|=O(n)\).

Now let
\[
\pi:\{1,\dots,d\}\to\{1,\dots,n\}
\]
be arbitrary. For each \(i\), let \(c_i^\pi\) be the minimal element of
the \((r+\pi(i))\)-th copy of \(U_i\) in \(\mathcal M_n\), and put
\[
\bar c^\pi=(c_1^\pi,\dots,c_d^\pi).
\]
By idempotency of the types \(\tau'_i\), the \(q\)-types of the
intervals determined by \(\bar c^\pi\) are the same as the
corresponding \(q\)-types of the intervals determined by \(\bar b\).
Hence
\[
\mathcal M_n\models D(\bar c^\pi),
\]
and therefore
\[
\mathcal M_n\models \exists \bar x\,G(\bar x,\bar c^\pi).
\]
It remains to choose the witnesses \(\bar x\) inside \(S^k\). Starting
from the tuple \(\bar a\) in \(\mathcal M_1\), transport each coordinate
lying in a copy of some \(U_i\) to the corresponding position in the
appropriate copy of \(U_i\) in \(\mathcal M_n\), relative to the shifted
point \(c_i^\pi\). Coordinates lying in some \(L_i\) are kept fixed.
The empty copies chosen on the left and on the right of the original
\(r\)-th copy serve as buffers, ensuring that this transport preserves
the \(q\)-types of all intervals determined by
\((\bar a,\bar b)\).

Thus we obtain a tuple \(\bar a^\pi\in S^k\) such that the relevant
intervals determined by \((\bar a^\pi,\bar c^\pi)\) in
\(\mathcal M_n\) have the same \(q\)-types as the corresponding
intervals determined by \((\bar a,\bar b)\) in \(\mathcal M_1\). By the
locality of $\MSO$ formulas on chains (\Cref{lem:local}), it follows that
\[
\mathcal M_n\models G(\bar a^\pi,\bar c^\pi).
\]
In particular,
\[
\mathcal M_n\models \varphi(\bar a^\pi).
\]
Different functions \(\pi\) give different tuples \(\bar c^\pi\).
Since \(G\) is functional, the corresponding tuples \(\bar a^\pi\) are
also distinct. Thus \(\varphi\) has at least \(n^d\) satisfying tuples
in \(S^k\). Since \(|S|=O(n)\), this yields
\[
g_\varphi(n)=\Omega(n^d).
\]